\pdfoutput=1
\RequirePackage{ifpdf}
\ifpdf % We are running pdfTeX in pdf mode
\documentclass[pdftex]{sigma}
\else
\documentclass{sigma}
\fi

\newcommand{\ZZ}{{\bf Z}}
\newcommand{\RR}{{\mathbb R}}
\newcommand{\CC}{{\mathbb C}}
\newcommand{\FF}{{\bf F}}
\newcommand{\HH}{{\mathbb H}}
\newcommand{\OO}{{\mathbb O}}
\newcommand{\PP}{{\bf P}}
\newcommand{\bS}{{\mathbb S}}
\newcommand{\cA}{{\cal A}}
\newcommand{\cS}{{\cal S}}
\newcommand{\tI}{{\tilde{I}}}
\newcommand{\tJ}{{\tilde{J}}}
\newcommand{\tS}{{\tilde{S}}}
\newcommand{\im}{\operatorname{im}}

\numberwithin{equation}{section}

\newtheorem{Theorem}{Theorem}[section]

\newtheorem{Lemma}[Theorem]{Lemma}
\newtheorem{Proposition}[Theorem]{Proposition}
{ \theoremstyle{definition}

\newtheorem{Example}[Theorem]{Example}
\newtheorem{Examples}[Theorem]{Examples}
\newtheorem{Remark}[Theorem]{Remark} }

\begin{document}
\allowdisplaybreaks

\newcommand{\arXivNumber}{1903.01228}

\renewcommand{\PaperNumber}{064}

\FirstPageHeading

\ShortArticleName{Lagrangian Grassmannians and Spinor Varieties in Characteristic Two}

\ArticleName{Lagrangian Grassmannians and Spinor Varieties\\ in Characteristic Two}

\Author{Bert VAN GEEMEN~$^\dag$ and Alessio MARRANI~$^{\ddag\S}$}

\AuthorNameForHeading{B.~van Geemen and A.~Marrani}

\Address{$^\dag$~Dipartimento di Matematica, Universit\`a di Milano, Via Saldini 50, I-20133 Milano, Italy}
\EmailD{\href{mailto:lambertus.vangeemen@unimi.it}{lambertus.vangeemen@unimi.it}}

\Address{$^\ddag$~Museo Storico della Fisica e Centro Studi e Ricerche Enrico Fermi,\\
\hphantom{$^\ddag$}~Via Panisperna 89A, I-00184, Roma, Italy}
\Address{$^\S$~Dipartimento di Fisica e Astronomia Galileo Galilei, Universit\`a di Padova,\\
\hphantom{$^\S$}~and INFN, sezione di Padova, Via Marzolo 8, I-35131 Padova, Italy}
\EmailD{\href{alessio.marrani@pd.infn.it}{alessio.marrani@pd.infn.it}}

\ArticleDates{Received March 08, 2019, in final form August 21, 2019; Published online August 27, 2019}

\Abstract{The vector space of symmetric matrices of size $n$ has a natural map to a projective space of dimension $2^n-1$ given by the principal minors. This map extends to the Lagrangian Grassmannian ${\rm LG}(n,2n)$ and over the complex numbers the image is defined, as a set, by quartic equations. In case the characteristic of the field is two, it was observed that, for $n=3,4$, the image is defined by quadrics. In this paper we show that this is the case for any $n$ and that moreover the image is the spinor variety associated to ${\rm Spin}(2n+1)$. Since some of the motivating examples are of interest in supergravity and in the black-hole/qubit correspondence, we conclude with a brief examination of other cases related to integral Freudenthal triple systems over integral cubic Jordan algebras.}


\Classification{14M17; 20G15; 51E25}

\section{Introduction}
In the paper \cite{HSL} the maximal commutative subgroups of the $n$-qubit Pauli group were studied. Such subgroups correspond to points in a Lagrangian Grassmannian ${\rm LG}(n,2n)$ over the Galois field $\FF_2$ with two elements. A subset of this Grassmannian is parametrized by symmetric $n\times n$ matrices. The principal minor map
\begin{gather*}
\underline{\pi}\colon \ {\cal S}_n := \{\text{symmetric}\ n\times n \ \text{matrices}\} \longrightarrow \PP\FF_2^{2^n} ,
\end{gather*}
which associates to a symmetric matrix with coefficients in $\FF_2$ its principal minors, extends to a map, again denoted by $\underline{\pi}$, on all of ${\rm LG}(n,2n)$:
\begin{gather*}
\underline{\pi}\colon \ {\rm LG}(n,2n) \longrightarrow Z_n\quad\big({\subset} \, \PP\FF_2^{2^n}\big),
\end{gather*}
where the image $Z_n$ of $\underline{\pi}$ is called the variety of principal minors of symmetric matrices.

Over the field of complex numbers, the variety $Z_n$ was studied in~\cite{HS}. In~\cite{Oeding} quartic equations which define~$Z_n$, as a set, were obtained. In case $n=3$, $Z_n$ is defined by a unique quartic polynomial which is Cayley's hyperdeterminant.

Returning to the case of the field $\FF_2$, it was observed that the hyperdeterminant reduces to the square of a quadratic polynomial over this field and~$Z_3$ is the quadric in $\PP\FF_2^3$ defined by this quadratic polynomial. Moreover, in~\cite{HSL} it was shown that for $n=4$ the variety $Z_n$ is defined by ten quadrics in~$\PP\FF_2^{16}$.

We will show that over any field of characteristic two, $Z_n$ is defined by quadrics for any $n\geq 3$. Moreover, these quadrics define the (image of the) well-known spinor variety $\bS_{n+1}$ associated to the group ${\rm Spin}(2n+1)$:
\begin{gather*}
Z_n \cong \bS_{n+1}\quad \big({\subset} \, \PP^{2^n-1}\big).
\end{gather*}

Over the complex numbers there is a natural embedding
\begin{gather*}
\sigma\colon \ \bS_{n+1} \longrightarrow \PP \CC^{2^n},
\end{gather*}
where $\CC^{2^n}$ is the spin representation of ${\rm Spin}(2n+1)$. Considering now a field of characteristic two, one obtains similarly an embedding
\begin{gather*}
\underline{\sigma}\colon \ \bS_{n+1} \longrightarrow \PP^{2^n-1}.
\end{gather*}
It is well-known that the image of $\underline{\sigma}$ (and of $\sigma$) is defined by quadrics. The spinor variety $\bS_{n+1}$ parametrizes maximally isotropic subspaces of a~smooth quadric. A subset of these subspaces is parametrized by alternating $(n+1)\times(n+1)$ matrices (since the characteristic is two, that means ${}^tA=-A$ (which is~$A$!) and all diagonal coefficients of $A$ should be zero). The maps $\sigma$ and $\underline{\sigma}$ are given by the $2^n$ Pfaffians of the principal submatrices of~$A$. The restriction of~$\underline{\sigma}$ to these subspaces will again be denoted by the same symbol:
\begin{gather*}
\underline{\sigma}\colon \ \cA_{n+1} := \{\text{alternating} \ (n+1)\times (n+1) \ \text{matrices} \} \longrightarrow \PP^{2^n-1}.
\end{gather*}

To show that $Z_n=\underline{\sigma}(\bS_{n+1})$, we will define in Section~\ref{mapsa}
an explicit map
\begin{gather*}
\alpha\colon \ \cS_n \longrightarrow \cA_{n+1},\qquad\text{such that}\qquad \underline{\pi}(S_n) = \underline{\sigma}(\alpha(S_n))
\end{gather*}
for all symmetric matrices $S_n\in \cS_n$ with coefficients in a(ny) algebraically closed field of characteristic~$2$. The proof involves an `induction on $n$' argument and the verification of a quadratic relation between the determinant of a symmetric matrix and certain of its principal minors, see Proposition~\ref{prop}.

The map $\alpha$ extends to a map
\begin{gather*}
\alpha\colon \ {\rm LG}(n,2n) \longrightarrow \bS_{n+1}.
\end{gather*}
To complete the picture, we discuss in Section~\ref{beta} a classical map, over fields of characteristic two,
\begin{gather*}
\beta\colon \ \bS_{n+1} \longrightarrow {\rm LG}(n,2n),\qquad \beta\alpha = F_{{\rm LG}(n,2n)},\qquad \alpha\beta = F_{\bS_{n+1}},
\end{gather*}
where $F$ is the Frobenius map, which is induced by the map $(\ldots:x_i:\ldots)\mapsto \big(\ldots:x_i^2:\ldots\big)$ on the projective spaces. This points to the `exceptional' isogeny of linear algebraic groups between ${\rm Spin}(2g+1)$ and ${\rm Sp}(2g)$ in characteristic two as the `reason' for these results.

In fact, after having completed a first draft of this paper, we became aware of the paper~\cite{Gow}, where R.~Gow uses this isogeny to provide the ingredients for a more intrinsic proof of the fact that $Z_n=\underline{\sigma}(\bS_{n+1})$ in characteristic two, see Remark~\ref{remarkisog}. We also noticed the recent paper \cite{LH} which involves the geometry studied in this paper.

Since the Cayley hyperdeterminant, ${\rm LG}(3,6)$ and $\bS_{6}$ also appear in the context of Freudenthal triple systems and four-dimensional Maxwell--Einstein supergravity on four space-time dimensions (as well as in \cite[Table~3]{Holweck}), we add a brief discussion on some characteristic two aspects of that topic.

\section{The maps}

\subsection{The fields}
Even if our motivation comes from algebra and geometry over the Galois field $\FF_2$ with two elements, we will consider the case of an algebraically closed field~$K$ of characteristic two. In such a field $2=0$ (and $-1=+1$), in particular the finite `binary' field $\FF_2=\ZZ/2\ZZ$ is contained in $K$, but $K$ will have infinitely many elements and one can do algebraic geometry over such a~field as well.

\subsection{Principal minors of symmetric matrices} \label{pm}
We recall the basics of the principal minors of a symmetric matrix. Let
\begin{gather*}
S_n := \left(\begin{matrix} x_{11}&x_{12}&x_{13}&\ldots&x_{1n}\\
x_{12}&x_{22}&x_{23}&\ldots&x_{2n}\\
x_{13}&x_{23}&x_{33}&\ldots&x_{3n}\\
\vdots&\vdots&\vdots&\vdots&\vdots\\
x_{1n}&x_{2n}&x_{3n}&\ldots&x_{nn}
\end{matrix}\right)~
\end{gather*}
be a symmetric $n\times n$ matrix. For a subset $I:=\{i_1,\ldots,i_k\}$ of $\{1,\ldots,n\}$ with $1\leq i_1<\dots<i_k\allowbreak \leq n$, the principal minor defined by $I$ is the determinant of the submatrix of $S_n$ with coefficients~$(S_n)_{ij}$ with $i,j\in I$. This principal minor will be denoted by $S_{n,I}$ and if $I$ is the empty set we put $S_{n,\varnothing}=1$. For example,
\begin{gather*}
S_{n,\varnothing}=1,\qquad
S_{n,\{i\}}=x_{ii},\qquad
S_{n,\{i,j\}}=x_{ii}x_{jj} - x_{ij}^2,\\
S_{n,\{i,j,k\}}=x_{ii}x_{jj}x_{kk}-x_{ii}x_{jk}^2 - x_{jj}x_{ik}^2 - x_{kk}x_{ij}^2 + 2x_{ij}x_{ik}x_{jk}.
\end{gather*}
The principal minor map $\pi\colon \cS_n\rightarrow \PP^{2^{n}-1}$ is defined by the
\begin{gather*}
2^{n} = \binom{n}{0} + \binom{n}{1} + \dots +\binom{n}{m} + \dots + \binom{n}{n}
\end{gather*}
principal minors of the $m\times m$ principal submatrices with $0\leq m\leq n$.

\begin{Example}[the case $n=3$] \label{n3}
Let $z_{abc}$, with $a,b,c\in\{0,1\}$ be the coordinates on $\PP^7$. The principal minor map
\begin{gather*}
\pi\colon \ \cS_3 \longrightarrow \PP^7,\qquad S_3 \longmapsto (z_{000}:z_{001}:\ldots:z_{111})= (S_{3,\varnothing}:S_{3,\{1\}}:\ldots:S_{3,\{1,2,3\}}),
\end{gather*}
in general, $z_{abc}=S_{3,I}$ where $1\in I$ iff $c=1$, $2\in I$ iff $b=1$ and $3\in I$ iff $a=1$. So $z_{100}=S_{3,\{3\}}=x_{33}$ and $z_{011}=S_{3,\{1,2\}}=x_{11}x_{22}-x_{12}^2$. The equation of the (Zariski closure of the) image of $\pi$ is $H=0$ where~$H$ is the hyperdeterminant~\cite{HS,Oeding})
\begin{gather*}
H :=z_{000}^2z_{111}^2+z_{001}^2z_{110}^2+z_{010}^2z_{101}^2+z_{100}^2z_{011}^2\\
\hphantom{H :=}{} -2(z_{000}z_{001}z_{110}z_{111} + z_{010}z_{011}z_{100}z_{101} + z_{000}z_{010}z_{101}z_{111} \\
\hphantom{H :=}{} +z_{001}z_{011}z_{100}z_{110} + z_{000}z_{011}z_{100}z_{111} + z_{001}z_{010}z_{101}z_{110})\\
\hphantom{H :=}{} +4(z_{000}z_{011}z_{101}z_{110}+z_{001}z_{010}z_{100}z_{111}).
\end{gather*}
In case we work over a field of characteristic two, $H$ is the square of a degree two polynomial
\begin{align*}
H & \equiv z_{000}^2z_{111}^2+z_{001}^2z_{110}^2+z_{010}^2z_{101}^2+z_{100}^2z_{011}^2\\
& \equiv (z_{000}z_{111}+z_{001}z_{110}+z_{010}z_{101}+z_{100}z_{011})^2 \quad\text{mod}\, 2,
\end{align*}
since now $(a+b)^2=a^2+2ab+b^2=a^2+b^2$. The (closure of the) image $Z_3$ of the map $\underline{\pi}\colon \cS_3\rightarrow\PP\FF_2^3$ is defined by this degree two polynomial, since
\begin{gather*}
1\cdot\big(x_{11}x_{22}x_{33}+x_{11}x_{23}^2 +x_{22}x_{13}^2 + x_{33}x_{12}^2\big)+ x_{11}\big(x_{22}x_{33}+x_{23}^2\big)\\
 \qquad{} + x_{22}\big(x_{11}x_{33}+x_{13}^2\big)+ x_{33}\big(x_{11}x_{22}+x_{12}^2\big)=0.
\end{gather*}
\end{Example}

\subsection{Pfaffians of alternating matrices} \label{Pfaff}
Let $A_N=(y_{ij})$ be the alternating $N\times N$ matrix where the coefficients $y_{ij}$ are the variables in the polynomial ring $R:=\ZZ[\ldots,y_{ij},\ldots]_{1\leq i<j\leq N}$
(so if $j>i$ then $y_{ji}=-y_{ij}$ and the diagonal coefficients of $A$ are zero):
\begin{gather*}
A = A_N :=\left(\begin{matrix}
0&y_{12}&y_{13}&y_{14}&\ldots&y_{1N}\\
-y_{12}&0&y_{23}&y_{24}&\ldots&y_{2N}\\
-y_{13}&-y_{23}&0&y_{34}&\ldots&y_{3N}\\
-y_{14}&-y_{24}&-y_{34}&0&\ldots&y_{4N}\\
\vdots&\vdots&\vdots&\vdots&\vdots&\vdots\\
-y_{1N}&-y_{2N}&-y_{3N}&-y_{4N}&\ldots&0
\end{matrix}\right).
\end{gather*}
Then $A$ corresponds to a $2$-form
\begin{gather*}
\sigma_A := \sum_{1\leq i<j\leq N} y_{ij}e_i\wedge e_j,
\end{gather*}
where the $e_i$ are the standard basis of $R^N$. In case $N$ is even, one defines a homogeneous polynomial $\operatorname{Pf}(A) \in \ZZ[\ldots,y_{ij},\ldots]_{1\leq i<j\leq N}$ of degree $N/2$ by considering the $N/2$-th exterior power of~$\sigma_A$:
\begin{gather*}
\sigma_A^{\wedge N/2} := \underbrace{\sigma_A\wedge\sigma_A\wedge \cdots\wedge\sigma_A}_{N/2} = {(N/2)!}\operatorname{Pf}(A)e_1\wedge\cdots\wedge e_N,
\end{gather*}
In case $N$ is odd, we simply put $\operatorname{Pf}(A)=0$.

For any field $K$ there is a natural homomorphism of rings $\ZZ\rightarrow K$ defined by sending $1\in\ZZ$ to $1\in K$ and this extends to a homomorphism of rings $\ZZ[\ldots,y_{ij},\ldots]\rightarrow K[\ldots,y_{ij},\ldots]$. The image $\operatorname{Pf}_K$ of the polynomial $\operatorname{Pf}(A)$ under this homomorphism defines the Pfaffian of an $N\times N$ alternating matrix coefficients in~$K$ as follows. Let $B=(b_{ij})$ be such an alternating matrix, then $\operatorname{Pf}(B):=\operatorname{Pf}_K(\ldots,b_{ij},\ldots)$, so we
evaluate $\operatorname{Pf}_K$ in $y_{ij}:=b_{ij}$.

It is not hard to verify the following formula for the Pfaffian of an alternating $N\times N$ mat\-rix~$A$ with coefficients~$y_{ij}$:
\begin{gather*}
\operatorname{Pf}(A) := \sum_{j=2}^N (-1)^jy_{1j}\operatorname{Pf}\big(A_{\hat{1}\hat{j}}\big),
\end{gather*}
where $A_{\hat{1}\hat{j}}$ is the $(N-2)\times(N-2)$ submatrix of $A$ where the first and $j$-th row and column of~$A$ are deleted. In case $\operatorname{char}(K)=2$ and $n$ is a fixed integer with $1\leq n\leq N$ one similarly has the following formula (we omit a sign since $\operatorname{char}(K)=2$ and notice that $y_{jj}=0$):
\begin{gather*}
\operatorname{Pf}(A) := \sum_{j=1}^N y_{jn}\operatorname{Pf}\big(A_{\hat{j}\hat{n}}\big),\qquad \operatorname{char}(K)=2.
\end{gather*}

For any subset $\tI\subset \{1,\ldots,N\}$ with an even number of elements we consider the `principal' submatrix of~$A$ with coefficients $(A_N)_{ab}$ and $a,b\in \tI$. These matrices are again alternating and thus we can consider their Pfaffians, which we denote by $A_{N,\tI}$ and we put $A_{N,\varnothing}:=1$. For example,
\begin{gather*}
A_{N,\varnothing}=1,\qquad A_{N,\{i,j\}}=y_{ij},\qquad A_{N,\{i,j,k,l\}}=y_{ij}y_{kl} - y_{ik}y_{jl} + y_{il}y_{jl}.
\end{gather*}
The Pfaffian map $\sigma\colon \cA_N\rightarrow \PP^{2^{N-1}-1}$ is defined by the
\begin{gather*}
2^{N-1} = \binom{N}{0} + \binom{N}{2} + \binom{N}{4} + \cdots
\end{gather*}
Pfaffians of the $m\times m$ principal submatrices, with $m$ even and $0\leq m\leq N$.

Since $e_i\wedge e_j$ and $e_k\wedge e_l$ commute in the exterior algebra $\wedge^*k^N$, one easily verifies that, with $A=A_N$,
\begin{gather*}
\exp(\sigma_A) := 1+\sigma_A+ \frac{1}{2!} \sigma_A\wedge \sigma_A+ \cdots + \frac{1}{(N/2)!}\sigma_A^{\wedge N/2} = \sum_\tI \operatorname{Pf}(A_\tI)e_\tI,
\end{gather*}
where the sum is over the ordered subsets $\tI=\{i_1,\ldots,i_{2k}\}\subset\{1,\ldots,N\}$ with an even number of elements and $e_\tI=e_{i_1}\wedge\cdots\wedge e_{i_{2k}}$,
since the $k!$ in the definition of $\operatorname{Pf}(A_\tI)$ cancels with the~$\frac{1}{k!}$ in the exponential function. Using commutativity as well as $(e_i\wedge e_j)^{\wedge 2} =(e_i\wedge e_j)\wedge(e_i\wedge e_j)=0$, we also have
\begin{gather*}
\exp(\sigma_A) = \exp\bigg(\sum_{i<j} y_{ij}e_i\wedge e_j\bigg) = \prod_{i<j}\exp(y_{ij}e_i\wedge e_j) =\prod_{i<j}(1+y_{ij}e_i\wedge e_j),
\end{gather*}
and thus
\begin{gather*}
\prod_{i<j}(1+y_{ij}e_i\wedge e_j) = \sum_\tI \operatorname{Pf}(A_\tI)e_\tI,
\end{gather*}
a formula which works over any field, also of finite characteristic. The Pfaffian map now appears as a natural map from $\cA_N$ into $\PP \wedge^{{\rm even}}K^N$.

\begin{Example}[the case $N=4$]\label{n4}
We define the Pfaffian map
\begin{gather*}
\sigma\colon \ \cA_4 \longrightarrow \PP^7,\qquad A_4 \longmapsto (z_{000}:z_{001}:\ldots:z_{111})= (A_{4,\varnothing}:A_{4,\{1,4\}}:\ldots:A_{4,\{1,2,3,4\}}),
\end{gather*}
by $z_{abc}=A_{4,\tilde{I}}$ and $\tI$ is obtained from $I$ with $z_{abc}=S_{3,I}$ in Example~\ref{n3} by $\tI=I$ if $\sharp I$, the cardinality of~$I$, is even and else $\tI=I\cup\{4\}$. One easily verifies that
\begin{gather*}
A_{4,\varnothing}A_{4,\{1,2,3,4\}} - A_{4,\{1,2\}}A_{4,\{3,4\}}+ A_{4,\{1,3\}}A_{4,\{2,4\}}-A_{4,\{1,4\}}A_{4,\{2,3\}} = 0,
\end{gather*}
hence the (closure of the image) of $\sigma$ is the quadric defined by $z_{000}z_{111}-z_{001}z_{110}+z_{010}z_{101}-z_{100}z_{011}$.

In particular, the image of $\underline{\sigma}\colon \cA_4\rightarrow\PP^{7}$ is defined by the degree two polynomial $z_{000}z_{111}+z_{001}z_{110}+z_{010}z_{101}+z_{100}z_{011}$ and thus, comparing with Example~\ref{n3}, the polynomials defining the images of $\underline{\sigma}$ (for $N=4$) and $\underline{\pi}$ (for $n=3$) are the same (and this holds over any field of characteristic two).
\end{Example}

\subsection{A map from symmetric to antisymmetric matrices}\label{mapsa}
In Example~\ref{n4} we observed that, over a field with characteristic two, the maps $\underline{\pi}$ and $\underline{\sigma}$, with domains $\cS_3$ and $\cA_4$ respectively, have images that are defined by the same quadratic polynomial. Now we define a map $\alpha\colon \cS_n\rightarrow\cA_{n+1}$ which will be shown
to have the property:
$\underline{\pi}(S_n)=\underline{\sigma}(\alpha(S_n))$ for any $n$.

With the notation from Sections \ref{pm} and \ref{Pfaff}, we define a (non-linear) map
\begin{gather*}
\alpha\colon \ \cS_n\longrightarrow \cA_{n+1},\qquad
S_n \longmapsto \tS_n:=\alpha(S_n),\\
\hphantom{\alpha\colon}{} \ \big(\tS_n\big)_{ij} =\begin{cases} 0 & \text{if }i=j,\\
 x_{ii}x_{jj}+x_{ij}^2 & \text{if } i \neq j, i,j\neq n+1,\\
 x_{ii} &\text{if }j=n+1,\\
 x_{jj} &\text{if }i=n+1.
 \end{cases}
\end{gather*}
Notice that we assume the field to have characteristic two, so $\tS_n$ is alternating (in fact, $\big(\tS_n\big)_{ii}=0$ for all $i$ and $\big(\tS_n\big)_{ij}=-\big(\tS_n\big)_{ji}=\big(\tS_n\big)_{ji}$).
For example,
\begin{gather*}
\alpha\colon \ S_3 =
\left(\begin{matrix} x_{11}&x_{12}&x_{13}\\
x_{12}&x_{22}&x_{33}\\
x_{13}&x_{23}&x_{33}
\end{matrix}\right)
\longmapsto \tS_3 =
\left(\begin{matrix} 0&x_{11}x_{22}+x_{12}^2&x_{11}x_{33}+x_{13}^2&x_{11}\\
x_{11}x_{22}+x_{12}^2&0& x_{22}x_{33}+x_{23}^2& x_{22}\\
x_{11}x_{33}+x_{13}^2& x_{22}x_{33}+x_{23}^2 &0&x_{33}\\
x_{11}&x_{22}&x_{33}&0
\end{matrix}\right).
\end{gather*}

Finally we define how the coordinate functions of $\underline{\pi}$ and $\underline{\sigma}$ correspond: for any subset $I\subset\{1,\ldots,n\}$ we define a subset $\tI\subset\{1,\ldots,n+1\}$ with an even number of elements as follows
\begin{gather*}
\tI = \begin{cases}
 I&\text{if} \ \sharp I \ \text{is even},\\
 I\cup\{n+1\}&\text{if} \ \sharp I \ \text{is odd}.
 \end{cases}
\end{gather*}
We will prove the following theorem in Section~\ref{secthm}:

\begin{Theorem} \label{thm} Let $\underline{\pi}\colon \cS_n\rightarrow\PP^{2^n-1}$ be the principal minor map with coordinate functions $S_{n,I}$ as in Section~{\rm \ref{pm}} and let $\underline{\sigma}\colon \cA_{n+1}\rightarrow\PP^{2^n-1}$ be the Pfaffian map with coordinate functions $A_{n+1,\tI}$ as in Section~{\rm \ref{Pfaff}} and where $I$ and $\tI$ correspond as above. Let $\alpha\colon \cS_n\rightarrow\cA_{n+1}$ be defined as in Section~{\rm \ref{mapsa}}.

Then we have, over any field of characteristic two
\begin{gather*}
\underline{\pi} = \underline{\sigma}\circ\alpha.
\end{gather*}
In fact, $S_{n,I}=\tS_{n,\tI}$ for all $S_n\in \cS_n$ and all subsets $I$ of $\{1,\ldots,n\}$.
\end{Theorem}

\begin{Examples}
We give some examples of the identity $S_{n,I}=\tS_{n,\tI}$. Obviously $S_{n,\varnothing}=1=\tS_{n,\varnothing}$. In case $I=\{i\}$ one has $\tI=\{i,n+1\}$ and indeed $S_{n,\{i\}}=x_{ii}=\tS_{n,\{1,n+1\}}$. In case $I=\{i,j\}$ one has $\tI=I$ and we do have the identity
\begin{gather*}
S_{n,\{i,j\}} = \det \left(\begin{matrix} x_{ii}&x_{ij}\\x_{ij}&x_{jj}\end{matrix}\right)
 = x_{ii}x_{jj}+x_{ij}^2 =
\operatorname{Pf} \left(\begin{matrix} 0&x_{ii}x_{jj}+x_{ij}^2\\x_{ii}x_{jj}+x_{ij}^2&0\end{matrix}\right) = \tS_{n,\{i,j\}}.
\end{gather*}

Finally if $I=\{i,j,k\}$ then $\tI=\{i,j,k,n+1\}$ and we do have $S_{n,\{i,j,k\}}=\operatorname{Pf}\big(\tS_{n,\{i,j,k,n+1\}}\big)$ because of the identity
\begin{gather*}
\det \left(\begin{matrix} x_{ii}&x_{ij}&x_{ik}\\
x_{ij}&x_{jj}&x_{kk}\\
x_{ik}&x_{jk}&x_{kk}
\end{matrix}\right) = \operatorname{Pf} \left(\begin{matrix}
0&x_{ii}x_{jj}+x_{ij}^2&x_{ii}x_{kk}+x_{ik}^2&x_{ii}\\
x_{ii}x_{jj}+x_{ij}^2&0& x_{jj}x_{kk}+x_{jk}^2& x_{jj}\\
x_{ii}x_{kk}+x_{ik}^2& x_{jj}x_{kk}+x_{jk}^2 &0&x_{kk}\\
x_{ii}&x_{jj}&x_{kk}&0
\end{matrix}\right),
\end{gather*}
which holds since
\begin{gather*}
x_{ii}x_{jj}x_{kk}+x_{ii}x_{jk}^2 + x_{jj}x_{ik}^2 + x_{kk}x_{ij}^2\\
\qquad{} = \big(x_{ii}x_{jj}+x_{ij}^2\big)x_{kk}+\big(x_{ii}x_{kk}+x_{ik}^2\big)x_{jj}+\big(x_{jj}x_{kk}+x_{jk}^2\big)x_{ii}.
\end{gather*}
Notice that these examples show that for $n=3$ we have $S_{3,I} = \tS_{3,\tI}$ for all subsets $I$ of $\{1,2,3\}$. Thus we verified Theorem~\ref{thm} for $n=3$ and this will be the starting point for an induction argument.
\end{Examples}

\section{The proof of Theorem \ref{thm}} \label{secthm}

\subsection{The determinant of a symmetric matrix}\label{termsn}
In order to prove Theorem~\ref{thm}, we start with some observations on the determinant of a symmetric matrix, in particular in the case the field has characteristic two.

The determinant of an $n\times n$ matrix $A=(a_{ij})$ is
\begin{gather*}
\det(A) = \sum_{\sigma\in\Sigma_n}\operatorname{sgn}(\sigma) a_{1\sigma(1)}\cdots a_{n\sigma(n)},
\end{gather*}
where $\Sigma_n$ is the symmetric group on $\{1,\ldots,n\}$. As $\det(A)=\det({}^tA)$, under the substitution $a_{ij}:=a_{ji}$ the monomials of the determinant are either fixed or permuted in pairs. A fixed term may contain any $a_{ii}$'s and if $a_{ij}$ occurs, so does $a_{ji}$. In a field of characteristic two, one has $+1=-1$ and $x+x=0$, so in a~determinant of a symmetric matrix over such a field the paired monomials will cancel and only the fixed monomials appear, all with coefficient~$1$. If $a_{ij}$, with $i\ne j$, occurs in a fixed term, then since $a_{ij}=a_{ji}$, the term contains~$a_{ij}^2$. Up to a simultaneous permutation of the rows and columns (to preserve the symmetry) any term in the determinant of the symmetric matrix $S_n$ is thus of the form
\begin{gather*}
x_{11}\cdots x_{kk}x_{k+1,k+2}^2\cdots x_{n-1,n}^2,\qquad k=0,1,\ldots, n .
\end{gather*}

\begin{Proposition}\label{prop} Let $K$ be a field of characteristic two and let $S_n=(x_{ij})$ be a symmetric $n\times n$ matrix. Then we have the following relation between
principal minors of~$S_n$:
\begin{enumerate}\itemsep=0pt
 \item[$(1)$] in case $n$ is even, %\label{even}
\begin{gather*}
\det(S_n) = \big(x_{11}x_{nn}+x_{1n}^2\big)\det(S_{n,\hat{1},\hat{n}})+\dots+\big(x_{n-1,n-1}x_{nn}+x_{n-1,n}^2\big)\det\big(S_{n,\widehat{n-1},\hat{n}}\big),
\end{gather*}
\item[$(2)$] in case $n$ is odd, %\label{odd}
\begin{gather*}
\det(S_n) = \big(x_{11}x_{nn}+x_{1n}^2\big)\det(S_{n,\hat{1},\hat{n}})+\dots +\big(x_{n-1,n-1}x_{nn}+x_{n-1,n}^2\big)\det\big(S_{n,\widehat{n-1},\hat{n}}\big)\\
\hphantom{\det(S_n) =}{} + x_{nn}\det(S_{n-1}),
\end{gather*}
\end{enumerate}
where $\det(S_{n,\hat{i},\hat{j}})$ is the principal minor $S_{n,I}$ with $I$ the subset of $\{1,\ldots,n\}$ with only $i$, $j$ omitted and $S_{n-1}=S_{n,\hat{n}}$ is the submatrix of $S_n$ where the last row and column are omitted.
\end{Proposition}

\begin{proof} The right hand sides of the two formulas in Proposition~\ref{prop} are invariant under simultaneous permutations of rows and columns which fix the last row and column. Therefore the formulas follow if the following monomials have equal coefficients on both sides of the identity
\begin{gather*}
t_k:=\big(x_{11}\cdots x_{kk}x_{k+1,k+2}^2\cdots\big)\cdot x_{n-1,n}^2,\qquad
t'_k:=\big(x_{11}\cdots x_{kk}x_{k+1,k+2}^2\cdots x_{n-2,n-1}^2\big)\cdot x_{nn}.
\end{gather*}
Notice that the $t_k$ appearing on the left hand side are those for which $n$ and $k$ have the same parity. Similarly, the $t'_k$ on the left are those for which $n$ and $k$ have different parity.

On the right hand side, each term in $(x_{ii}x_{nn}+x_{in}^2)\det(S_{n,\widehat{i},\hat{n}})$ and also in $x_{nn}\det(S_{n-1})$ is a~$t_k$ or a~$t_k'$ up to simultaneous permutation of rows and columns. So we only need to verify that each term of type~$t_k$ occurs an odd number of times in the summands on the right hand sides of Proposition~\ref{prop}.

The terms $t_k$ all have the variable $x_{n-1,n}$. In the matrices $S_{n,\hat{i},\hat{n}}$ ($i=1,\ldots,n-1$) and $S_{n-1}$ appearing in the two formulas in Proposition~\ref{prop} we omit the $n$-th row and column, so they don't have the variable $x_{n-1,n}$. Only $x_{n-1,n-1}x_{nn}+x_{n-1 n}^2$ has this variable. Each $t_k$, $k=0,\ldots, n-2$, thus occurs at most once in the expansion of the right hand side. It is also not hard to see that each $t_k$ actually occurs in $x_{n-1 n}^2\cdot \det(S_{n,\widehat{n-1},\hat{n}})$, provided $k$ has the same parity as $n$.

Now consider the terms $t'_k$. We notice first of all that $t'_{n-1}=x_{11}\cdots x_{nn}$ occurs in all terms on the right hand side of each of the two formulas in Proposition~\ref{prop} and since the two right hand sides each have an odd number of terms, it survives.

Next we consider $t'_{n-3}=x_{11}\cdots x_{n-3,n-3} x_{n-2,n-1}^2 x_{n,n}$. Considering $x_{n-2,n-1}^2$, it obviously does not occur in the two terms
\begin{gather*}
\big(x_{n-2,n-2}x_{nn}+x_{n-2,n}^2\big)\det\big(S_{n,\widehat{n-2},\hat{n}}\big),\qquad
\big(x_{n-1,n-1}x_{nn}+x_{n-1,n}^2\big)\det\big(S_{n,\widehat{n-1},\hat{n}}\big).
\end{gather*}
However, $t'_{n-3}$ does appear in all other summands of each of the two right hand sides in Proposition~\ref{prop}. Thus $t'_{n-3}$ appears in an odd number of summand and hence it appears on the right hand side. More generally, $t'_{n-2k}$ does not appear in the $2k$ summands $\big(x_{n-i,n-i}x_{nn}+ x_{n-i,n}^2\big)\det\big(S_{n,\widehat{n-i},\hat{n}}\big)$ for $i=1,\ldots,2k$, but it appears in all other summands. Hence $t'_{n-2k}$ appears in an odd number of summands and hence it appears on the right hand side. This concludes the proof of Proposition~\ref{prop}.
\end{proof}

\begin{proof}[Proof of Theorem \ref{thm}] %\label{proofThm}
We need to show that $S_{n,I} = \tS_{n,\tI}$ for any $n$ and any $I \subset\{1,\ldots,n\}$. We proceed by induction on $n$, and we already verified the equalities for all $I$ in the case $n=3$. So we assume that $S_{n,I} = \tS_{n,\tI}$ holds for all $I\subset \{1,\ldots,n\}$ and we must prove that $S_{n+1,J} = \tS_{n+1,\tJ}$ for all subsets $J\subset\{1,\ldots,n+1\}$.

In case $\sharp J<n+1$, after a permutation of the indices, we may assume that $J=\{1,2,\ldots,k\}\subset\{1,\ldots,n\}$, and then $S_{n+1,J} = \tS_{n+1,\tJ}$ follows from the induction hypothesis. To deal with the remaining case $J=\{1,\ldots,n+1\}$ we distinguish the cases $n+1$ odd and $n+1$ even.

In case $n+1$ is odd, $\tilde{J}=\{1,\ldots,n+1,n+2\}$ and we must show that $S_{n+1,{J}}=\tS_{n+1,\tilde{J}}$, that is $\det(S_{n+1})=\operatorname{Pf}\big(\tS_{n+1}\big)$. It is more convenient to change the integer $n$ to $n-1$ and then we must show $\det(S_{n})=\operatorname{Pf}\big(\tS_{n}\big)$ for~$n$ odd. Using the formula for computing the Pfaffian given in Section~\ref{Pfaff} (with $N=n+1$) we have
\begin{gather*}
\operatorname{Pf}\big(\tS_{n}\big) = \sum_{k=1}^{n+1}\big(\tS_n\big)_{k,n}\operatorname{Pf}\big(\tS_{n,\hat{k},\hat{n}}\big)
 = \left(\sum_{k=1}^{n-1} (x_{kk}x_{nn}+x_{kn}^2)\operatorname{Pf}\big(\tS_{n,\hat{k},\hat{n}}\big)\right)
 + x_{nn}\operatorname{Pf}\big(\tS_{n,\hat{n},\widehat{n+1}}\big).
\end{gather*}
The principal submatrix $\tS_{n,\hat{n}}$ of $\tS_n$ obtained by deleting the $n$-th row and column, is an alterna\-ting $n\times n$ matrix where the coefficients $x_{in}$ no longer appear and which is exactly~$\tS_{n-1}$, so \smash{$\tS_{n,\hat{n}}=\tS_{n-1}$}. For all $k\in \{1,\ldots,n-1\}$ the Pfaffian of the $(n-1)\times (n-1)$ alter\-na\-ting matrix $\tS_{n-1,\hat{k}}$ obtained by deleting the $k$-th row and column of $\tS_{n-1}$ is $\tS_{n-1,\tI}$ where $\tI=\big\{1,\ldots,\hat{k},\ldots,n\big\}$. By induction we know that this Pfaffian is $\det(S_{n-1,I})$ where $I=\big\{1,\ldots,\hat{k},\ldots,n-1\big\}$ in case $k<n$, which is also $\det\big(S_{n-1,\hat{k}}\big)$. In case $k=n$, we have $\big(\tS_n\big)_{n,n}=0$ and we already omitted this term. Finally if $k=n+1$ we have $\tS_{n,\hat{n},\widehat{n+1}}=\tS_{n-1,\tI}$ where $\tI=\{1,\ldots,n-1\}$ and thus, by induction, $\operatorname{Pf}\big(\tS_{n,\hat{n},\widehat{n+1}}\big)=\det(S_{n-1})$. Thus we can rewrite the Pfaffian of $\tS_n$ in terms of principal minors of $S_{n-1}$:
\begin{gather*}
\operatorname{Pf}\big(\tS_{n}\big) =\left(\sum_{k=1}^{n-1} \big(x_{kk}x_{nn}+x_{kn}^2\big)\det\big(S_{n-1,\hat{k}}\big)\right) + x_{nn}\det(S_{n-1}),
\end{gather*}
and the equality $\det(S_n)=\operatorname{Pf}\big(\tS_{n}\big)$ for $n$ odd follows from Proposition~\ref{prop}(2).

In case $n+1$ is even, $J=\{1,\ldots,n+1\}=\tJ$ and we must show that $S_{n+1,{J}}=\tS_{n+1,{J}}$, that is $\det(S_{n+1})=\operatorname{Pf}\big(\tS_{n+1,\widehat{n+2}}\big)$. Again we prefer to change the integer~$n$ to $n-1$, so we must show that for $n$ even we have $\det(S_{n})=\operatorname{Pf}\big(\tS_{n,\widehat{n+1}}\big)$. We have the following expansion of the Pfaffian of the alternating $n\times n$ matrix $\tS_{n,\widehat{n+1}}$:
\begin{gather*}
\operatorname{Pf}\big(\tS_{n,\widehat{n+1}}\big) =
\sum_{k=1}^{n-1}\big(\tS_{n,\widehat{n+1}}\big)_{k,n}\operatorname{Pf}\big(\tS_{n,\hat{k},\hat{n},\widehat{n+1}}\big) =
\sum_{k=1}^{n-1} \big(x_{kk}x_{nn}+x_{kn}^2\big)\operatorname{Pf}\big(\tS_{n,\hat{k},\hat{n},\widehat{n+1}}\big).
\end{gather*}
Notice that $\tS_{n,\hat{k},\hat{n},\widehat{n+1}}=\tS_{n-1,\hat{k},\hat{n}}$ and by induction we may assume that
\begin{gather*} \operatorname{Pf}\big(\tS_{n-1,\hat{k},\hat{n}}\big)= \det\big(S_{n-1,\hat{k}}\big),
\end{gather*} since if $n$ is even, then $I:=\big\{1,\ldots,\hat{k},\ldots,n-1\big\}=\tI$. Finally we notice that $S_{n-1,\hat{k}}=S_{n,\hat{k},\hat{n}}$. Thus the equality $\det(S_{n})=\operatorname{Pf}\big(\tS_{n,\widehat{n+1}}\big)$ for $n$ even follows from Proposition~\ref{prop}(1).
\end{proof}

\section{From matrices to Grassmannians}

\subsection{Global aspects}
We recall that the spaces of symmetric and antisymmetric matrices have a natural interpretation as open subsets of certain Grassmannians, like the spinor varieties, and that the principal minor map $\pi$ and the Pfaffian map $\sigma$ extend to these Grassmannians. We also discuss the actions of some groups on these Grassmannians. In the final section we recall that the image of the spinor variety is defined by quadrics.

\subsection{The Lagrangian Grassmannian} \label{Leg}
Let $V$ be a vector space over a field $K$ and let
\begin{gather*}
e\colon \ V\times V \longrightarrow K,
\end{gather*}
be a symplectic form, that is, an alternating, non-degenerate, bilinear form (so for any $x\in V$, $e(x,x)=0$ and if $x\ne 0$, there is a $y\in V$ with $e(x,y)\ne 0$). Then $V$ has a~symplectic basis $f_1,\ldots,f_{2n}$, that is, $e(f_i,f_{j+n})=-e(f_{j+n},f_i)=\delta_{ij}$ (Kronecker's delta) for $1\leq i,j\leq n$ and all other $e(f_i,f_j)$ are zero. So if ${\bf I}$ denotes the $n\times n$ identity matrix, then
\begin{gather*}
e\left(\sum_{i=1}^{2n} x_if_i,\sum_{j=1}^{2n}y_jf_j\right) = \sum_{i=1}^n x_iy_{i+n} - x_{i+n}y_i = (x_1\ldots x_{2n})
\begin{pmatrix}
 0&{\bf I}\\-{\bf I}&0
 \end{pmatrix}\begin{pmatrix}
 y_1\\ \vdots\\y_{2n}
 \end{pmatrix}.
\end{gather*}
A (linear) subspace $W\subset V$ is called isotropic if $e(w,w')=0$ for all $w,w'\in W$ and $W$ is called Lagrangian if it is isotropic and $\dim W=n$, the maximal possible. Choosing a basis $w_1,\ldots, w_n$ of $W$, let $M_W$ be the $2n\times n$ matrix whose columns are the $w_i$. Then $W=\operatorname{im}\big(M_W\colon K^n\rightarrow K^{2n}\big)$ and $W$ is isotropic iff
\begin{gather*}
{}^tM_W\begin{pmatrix}
 0&{\bf I}\\-{\bf I}&0
 \end{pmatrix} M_W = 0\quad\Longleftrightarrow\quad
 {}^tAB - {}^tBA = 0,\qquad M_W = \begin{pmatrix} A\\B\end{pmatrix}.
\end{gather*}
In particular, the subspace $W_0:=\langle f_1,\ldots,f_n\rangle $ is Lagrangian and $M_{W_0}$ has blocks $A={\bf I}$ and $B=0$. More generally, given a symmetric $n\times n$ matrix $X$, the subspace $W_X$ spanned by the columns of the matrix $M$ with blocks $A=I$ and $B=X$ is Lagrangian
\begin{gather*}
\cS_n \hookrightarrow {\rm LG}(n,2n),\qquad X \longmapsto W_X := \operatorname{im}\begin{pmatrix} {\bf I}\\X\end{pmatrix}.
\end{gather*}
The Lagrangian subspaces of $K^{2n}$ are parametrized by the Lagrangian Grassmannian ${\rm LG}(n,2n)$, an algebraic subvariety of dimension $n(n+1)/2$ of the Grassmannian ${\rm Gr}(n,2n)$ of all $n$-di\-men\-sio\-nal subspaces of $K^{2n}$.

\subsection{The Pl\"ucker map}\label{pluc}
The Pl\"ucker map gives an embedding of
\begin{gather*}
{\rm Gr}(n,2n) \longrightarrow \PP \wedge^n K^{2n}, \qquad W \longmapsto \wedge^nW = \sum_I p_I(W)f_I,
\end{gather*}
where $I=\{i_1,\ldots,i_n\}$ is an ordered subset of $\{1,\ldots,2n\}$ and $f_I:=f_{i_1}\wedge\dots\wedge f_{i_n}$ where the~$f_i$ are the standard basis of $K^{2n}$. If $W$ is the span of the columns of an $2n\times n$ matrix $M_W$, then~$p_I(W)$ is the determinant of the $n\times n$ submatrix of~$M_W$ given by the rows $i_1,\ldots,i_n$ of~$M_W$.

To understand the restriction of the Pl\"ucker map to the submanifold ${\rm LG}(n,2n)$ of ${\rm Gr}(n,2n)$, we recall some general results on the exterior algebra of a symplectic vector space over a~field~$K$ of characteristic zero (see~\cite{Gow} and the references given there, or \cite[Section~11.6.7]{P}, but note the misprints). Let $e$ be the standard symplectic form on $V:=K^{2n}$, then one defines contraction maps
\begin{gather*}
\partial\colon \ \wedge^kV \longrightarrow \wedge^{k-2} V,\\
\hphantom{\partial\colon}{} \ \partial(v_1\wedge\cdots\wedge v_k) := \sum_{i<j}e(v_i,v_j)(-1)^{i+j-1}
v_1\wedge\cdots\wedge \widehat{v_i}\wedge\cdots\wedge\widehat{v_j}\wedge\cdots\wedge v_k.
\end{gather*}
Let the $f_i$ be a symplectic basis of $V$ as before, then we define
\begin{gather*}
\epsilon\colon \ \wedge^k V \longrightarrow \wedge^{k+2}V,\qquad \theta\longmapsto \Gamma\wedge\theta
\qquad\text{with}\qquad \Gamma := \sum_{i=1}^nf_i\wedge f_{i+n}\quad\big({\in}\, \wedge^2 V\big).
\end{gather*}
We extend $\partial$ and $\epsilon$ to the exterior algebra $\wedge^*V$ of $V$ by linearity. Finally we define a linear map
\begin{gather*}
H\colon \ \wedge^*V := \bigoplus_{k=0}^{2n}\wedge^kV \longrightarrow \wedge^*V,\qquad
H(\theta) = (n-k)\theta\qquad\text{if}\quad\theta\in\wedge^kV.
\end{gather*}
These linear maps define a representation of the Lie algebra ${\mathfrak{sl}}(2)$ on $\wedge^*V$:
\begin{gather*}
H = [\partial,\epsilon],\qquad [H,\partial] = 2\partial,\qquad [H,\epsilon]=-2\epsilon.
\end{gather*}
We denote the subspace of highest weight vectors, of weight $n-k\geq 0$, for this ${\mathfrak{sl}}(2)$-representation by
\begin{gather*}
\big({\wedge}^kV\big)_0 := \big\{ \theta\in\wedge^kV\colon \partial\theta = 0 \big\},\qquad k = 0,1,\ldots,n.
\end{gather*}
As a consequence, there is a decomposition (\cite[Section~11.6.7, Theorem~3]{P}, basically the Lefschetz decomposition from \cite[p.~122]{GH}),
\begin{gather*}
\wedge^kV = \bigoplus_{2i\geq k-n} \Gamma^i\wedge\big({\wedge}^{k-2i}V\big)_0,
\end{gather*}
which is the decomposition of $\wedge^kV$ into irreducible ${\rm Sp}(2n)$ subrepresentations. In the case $k=n$, the vector space $\wedge^nV$ is the weight space for ${\mathfrak{sl}}(2)$ with weight $0$, and thus
\begin{gather*}
\wedge^nV = \big({\wedge}^nV\big)_0\oplus V'_n, \qquad V'_n=\im\big(\epsilon\colon \wedge^{n-2}V\hookrightarrow\wedge^nV\big) = \im\big(\partial\colon \wedge^{n+2}V\hookrightarrow\wedge^{n}V\big),
\end{gather*}
and $\big({\wedge}^nV\big)_0$ is a trivial ${\mathfrak{sl}}(2)$-representation, moreover, $\epsilon^2\colon \wedge^{n-2}V\rightarrow\wedge^{n+2}V$, $\partial^2\colon \wedge^{n+2}V\rightarrow\wedge^{n-2}V$ are isomorphisms.

Let $W$ be a Lagrangian subspace of $V$. Then one can choose a symplectic basis $f_i$ for $V$ such that $f_1,\ldots,f_n$ are a basis of $W$ and one easily finds that now
$\Gamma\wedge\big({\wedge}^nW\big)=0\in \wedge^{n+2}V$. Since the decomposition of $\wedge^nV$ does not depend on the choice of a symplectic basis we find that
\begin{gather*}
{\rm LG}(n,2n) = {\rm Gr}(n,2n) \cap \PP\big({\wedge}^n V\big)_0\qquad\big({\subset} \,\PP \wedge^n V\big),
\end{gather*}
where we view ${\rm Gr}(n,2n)$ as a subvariety of $\PP\big({\wedge}^n V\big)$.

For example, if $n=3$ then ${\rm LG}(3,6)$ maps to $\PP^{13}$ since the dimension of $\big({\wedge}^3V\big)_0$ is then $20-6=14$, this case is discussed in~\cite{IR} and Section~\ref{varcas}.

\subsection{The principal minor map}\label{pmm}
The principal minor map extends to a map, again denoted by $\pi$,
\begin{gather*}
\pi\colon \ {\rm LG}(n,2n) \longrightarrow \PP^{2^n-1},\qquad W \longmapsto (\ldots:p_J(W):\ldots),
\end{gather*}
where $J$ runs over the $2^n$ special subsets $J\subset\{1,\ldots,2n\}$ with $\sharp J=n$, where, for every $i\in \{1,\ldots,n\}$, $J$ contains either $i$ or $n+i$. In case $W$ is the image of $M_W$ and $M_W$ has blocks ${\bf I}$ and $X\in\cS_n$, then these $p_J(W)$ are easily seen to be the principal minors of~$X$. Thus $\pi$ is a projection
of ${\rm LG}(n,2n)\subset \PP\big({\wedge}^nK^{2n}\big)_0$ into $\PP^{2^n-1}$ and it is not hard to verify that~$\pi$ is a regular map (base point free) on ${\rm LG}(n,2n)$. The closure $Z_n$ of $\pi(\cS_n)$ is thus the projective variety $\pi({\rm LG}(n,2n))$.

We now show that the morphism $\pi\colon {\rm LG}(n,2n)\rightarrow Z_n$ has degree $2^{n-1}$, if the characteristic of the field~$K$ is not two. (In the lemma below, ${\rm LG}(n,2n)/G_n$ is not isomorphic to $Z_n$ for $n>3$ since there are invariant monomials in the $x_{ij}$ on $\cS_n\subset {\rm LG}(n,2n)$ which are not contained in the ring of principal minors.)

\begin{Lemma}The principal minor map $\pi\colon {\rm LG}(n,2n)\rightarrow Z_n \big({\subset}\,\PP^{2^n-1}\big)$ has degree $2^{n-1}$ over a~field of characteristic different from two. This map factors over a~quotient of ${\rm LG}(n,2n)$ by a~group $G_n\cong (\ZZ/2\ZZ)^{n-1}$.
\end{Lemma}

\begin{proof} Any diagonal matrix $D=\operatorname{diag}(t_1,\ldots,t_n,t_1^{-1},\ldots,t_n^{-1})$ with $t_i\ne 0$ fixes the symplectic form $e$ and thus maps ${\rm LG}(n,2n)$ into itself by $W\mapsto DW$, equivalently, $M_W\mapsto DM_W$. Let $D_1:=\operatorname{diag}(t_1,\ldots,t_n)$, and notice that $DM_W$ and $DM_WD_1^{-1}$ map $K^n$ to the same subspace~$DW$ in~$K^{2n}$. For $M_W$ with blocks ${\bf I},X$, the matrix $DM_WD_1^{-1}$ has blocks ${\bf I}$, $D_1^{-1}XD_1^{-1}$, so we see that $D$ maps the image of $\cS_n$ in ${\rm LG}(n,2n)$ into itself and acts as $D\colon X\mapsto D_1^{-1}XD_1^{-1}$. In case all $t_i\in\{1,-1\}$, we have $D_1^{-1}=D_1$ and we write more suggestively $D\colon X\mapsto D_1XD_1^{-1}$, the conjugation by $D_1$. Any principal submatrix of~$X$ is then also conjugated by a submatrix of~$D_1$, and hence the principal minors of $X$ and those of $D_1XD_1^{-1}$ are the same. So the fiber of~$\pi$ over~$\pi(X)$ contains all the $D_1XD_1^{-1}$ where $D_1$ has coefficients~$\pm 1$. Obviously $D_1=-I$ acts trivially and thus we have an action of the group $G_n:=(\ZZ/2\ZZ)^{n-1}$ on ${\rm LG}(n,2n)$ and $\pi$ factors over~${\rm LG}(n,2n)/G_n$. The $ij$-coefficient of $D_1XD_1^{-1}$ is $x_{ij}t_it_j$. Since the $x_{ii},x_{ii}x_{jj}-x_{ij}^2$ are principal minors of $X$, we can recover the $x_{ij}$ from $\pi(W_X)$, except for the signs of the $x_{ij}$ with $i\ne j$. However, the principal minors $S_{n,\{i,j,k\}}$ (see Section~\ref{pm}) show that once, for a fixed~$i$, all the $x_{il}$ are non-zero and the signs of all these $x_{il}$ are fixed, then the signs of all $x_{jk}$ are fixed. Therefore the fiber over $\pi(X)$, for general $X\in \cS_n$, consists of exactly $2^{n-1}$ elements that are an orbit of $G_n$. This implies that~$\pi$ has degree~$2^{n-1}$ and that~$\pi$ factors over~${\rm LG}(n,2n)/G_n$.
\end{proof}

\subsection{The spinor varieties}\label{isotr}\label{pfmap}
A quadratic form on a vector space $V$ over a field $K$ is a map
\begin{gather*}
q\colon \ V \longrightarrow K,\qquad \text{such that}\qquad q(ax) = a^2q(x),\qquad q(x+y) = q(x) + q(y) +e(x,y),
\end{gather*}
where $a\in K$ and $e$ is a bilinear form and $x,y\in V$. We consider the quadratic form $q$ on $V=K^{2n}$ defined by
\begin{gather*}
q\left(\sum_{i=1}^{2n}x_if_i\right) := \sum_{i=1}^nx_ix_{i+n},\qquad 2q(x) = (x_1\ldots x_{2n})
\begin{pmatrix}
 0&{\bf I}\\{\bf I}&0
 \end{pmatrix}\begin{pmatrix}
 x_1\\ \vdots\\x_{2n}
 \end{pmatrix}.
\end{gather*}
A (linear) subspace $W\subset V$ is called an isotropic subspace of $q$ if $q(w)=0$ for all $w\in W$ and it is a maximally isotropic subspace of $q$ if moreover $\dim W=n$, the maximum possible. Choosing a basis $w_1,\ldots, w_n$ of $W$, let $M_W$ be the $2n\times n$ matrix whose columns are the $w_i$. Then $W=\operatorname{im}\big(M_W\colon K^n\rightarrow K^{2n}\big)$. The subspace $W$ is maximally isotropic for $q$ iff
\begin{gather*}
q(w_i) = 0,\qquad q(w_i+w_j) = 0,\qquad 1\le i,j\leq n,
\end{gather*}
in fact, if $q(w_i)=0$ and also $0=q(w_i+w_j)=e(w_i,w_j)$ for all $i$, $j$, then from
\begin{align*}
q\left(\sum_{i=1}^{n}a_iw_i\right) & = q\left(\sum_{i=1}^{n-1}a_iw_i\right)+a_n^2q(w_{n})+\sum_{i=1}^{n-1}a_ia_{n}e(w_i,w_{n}) \\
& = \sum_{i=1}^na_i^2q(w_i) + \sum_{i<j} a_ia_je(w_i,w_j)
\end{align*}
we see that $W$ is maximally isotropic. In case $\operatorname{char}(K)\ne 2$ this can also be checked using the symmetric matrix of $e$:
\begin{gather*}
{}^tM_W\begin{pmatrix}
 0&I\\I&0
 \end{pmatrix} M_W = 0\quad\Longleftrightarrow\quad
 {}^tAB + {}^tBA = 0,\qquad M_W = \begin{pmatrix} A\\B\end{pmatrix},
\end{gather*}
and notice that $q(w_i)=\big({}^tAB + {}^tBA\big)_{ii}$ and $q(w_i +w_j)=\big({}^tAB + {}^tBA\big)_{ij}$.

The subspace $W_0:=\langle f_1,\ldots,f_n\rangle$ is thus maximally isotropic for $q$. More generally, given an antisymmetric $n\times n$ matrix~$Y$, the subspace $W_Y$ spanned by the columns of the matrix $M$ with blocks $A={\bf I}$ and $B=Y$ is Lagrangian, so
\begin{gather*}
\cA_n \hookrightarrow \bS_n^+,\qquad Y \longmapsto \operatorname{im} \begin{pmatrix} {\bf I}\\Y\end{pmatrix},
\end{gather*}
where $\bS_n^+$ denotes the spinor variety containing $W_0$. This holds over any field, since $q(w_i)=y_{ii}$ and $q(w_i+w_j)=y_{ii}+y_{jj}+y_{ij}+y_{ji}$ and thus $W_Y$ is maximally isotropic for $q$ iff the diagonal coefficients of $Y$ are zero and $y_{ij}+y_{ji}=0$ iff $Y$ is alternating. Recall that there are two $n(n-1)/2$-dimensional families of maximally isotropic subspaces of $q$. They are parametrized by the spinor varieties $\bS_{n}^+$ and $\bS_{n}^-$, which are isomorphic. For spinor varieties see \cite{C}, \cite[Section~11.7]{P} and the references given in \cite[Section~6.0]{RS}.

\subsection{The image of the Pfaffian map}\label{impfaff}
Over the complex numbers, the Pfaffian map on $\cA_n$ from Section~\ref{Pfaff} extends to an embedding of the spinor variety
\begin{gather*}
\sigma\colon \ \bS_n^+ \longrightarrow \PP^{2^{n-1}-1}.
\end{gather*}
In the introduction we used a map $\sigma$ on the spinor variety associated to ${\rm Spin}(2n-1)$, but we will see that these spinor varieties are isomorphic in Section~\ref{evenodd}.

The spinor variety $\bS_{n}^+$ is the homogeneous variety $G/P$, with $G={\rm Spin}(2n)$ and the image of~$\sigma$ consists of the pure spinors (for any one of the two half spin representations of $G$), as in \cite[Section~11.7.2]{P}, $\sigma(\bS_n)$ is also the $G$-orbit of the highest weight vector in the projectivization of the half spin representation. Under certain natural identifications, the Lie algebra of the Spin group is identified with a subspace of the Clifford algebra $C(q)$ of~$q$ and a maximally isotropic subspace~$W$ of~$q$ defines a subalgebra $\wedge^*W\subset C(q)$. In case $e_1,\ldots,e_n$ is a basis of $W$, the element $\exp(y_{ij}e_i\wedge e_j)=\prod(1+y_{ij}e_i\wedge e_j)$ introduced in Section~\ref{Pfaff} is actually an element of the Spin group and from this one can deduce that the orbit of the highest weight vector is indeed locally parametrized by the Pfaffian map.

In general, the orbit under a semisimple simply connected algebraic group~$G$ (defined over an algebraically closed field of arbitrary characteristic) of a highest weight vector in an irreducible minuscule representation of~$G$ is the intersection of quadrics, see~\cite{S}. This implies that the image of $\sigma$ is an intersection of quadrics. The number of quadrics can also be determined, it is
\begin{gather*}
\dim I_2 := \binom{2^{n-1}+1}{2} - \frac{1}{2}\binom{2n}{n},\qquad
I_2 := \big\{Q\in k[\ldots,z_I,\ldots]\colon Q(\sigma(W)) = 0\ \forall\, W\in \bS_n^+ \big\},
\end{gather*}
where $K[\ldots,z_I,\ldots]$ is the homogeneous coordinate ring of $\PP^{2^{n-1}-1}$, in fact, \cite{S} shows that $\dim I_2$ does not depend on the characteristic of the field
and over the complex numbers one can use for example (the proof of) \cite[Theorem~2]{vGspin}). So for $n=4,5,6$ we find $36-35=1$, \mbox{$136-126=10$}, $528-462=66$ quadrics respectively. See also the end of section \cite[Section~11.7.2]{P} for the quadratic relations between Pfaffians, \cite{SV} for explicit methods to find the quadratic equations of $\sigma\big(\bS_{n}^+\big)$ and \cite[Section~6]{RS} for a study of the case $n=5$.

\begin{Proposition}\label{quadr} Let $\underline{\pi}\colon \cS_n \rightarrow \PP^{2^n -1}$ be the principal minor map over an algebraically closed field of characteristic two. Then the closure~$Z_n$ of the image of $\underline{\pi}$ is $\underline{\sigma}\big(\bS_{n+1}^+\big)$ and in particular~$Z_n$ is an intersection of quadrics.
\end{Proposition}

\begin{proof} Since the symmetric matrices $\cS_n$ are Zariski dense in ${\rm LG}(n,2n)$ and the alternating matrices are Zariski dense in $\bS_{n+1}^+$, we find, using Theorem~\ref{thm}, that
\begin{gather*}
Z_n = \underline{\pi}({\rm LG}(n,2n)) = \underline{\sigma}\big(\bS_{n+1}^+\big) \subset\PP^{2^n-1}.
\end{gather*}
In Section \ref{impfaff} we recalled that $\underline{\sigma}\big(\bS_{n+1}^+\big)$ is defined by quadrics, hence also $Z_n$ is defined by quadrics.
\end{proof}

\section[The map $\beta$]{The map $\boldsymbol{\beta}$}\label{beta}

\subsection{From antisymmetric to symmetric matrices}\label{mapas}
\looseness=1  We work over a field of characteristic two. In Section~\ref{mapsa} we defined $\alpha\colon \cS_n\rightarrow\cA_{n+1}$ in such a~way that the principal minors of $S_n$ were the Pfaffians of $\alpha(S_n)$, this condition determined the map~$\alpha$. Now we consider a map $\beta\colon \cA_{n+1}\rightarrow\cS_{n}$, which is defined in terms of a~well-known map from $\bS_{n+1}^+\cong\bS_{n+1}\rightarrow {\rm LG}(n,2n)$, which we will also denote by~$\beta$. The maps $\alpha$ and $\beta$ are not mutual inverses, instead their compositions are purely inseparable maps, given by squaring all coefficients in the matrix. Since the field has characteristic two, these maps are injective and if the field is algebraically closed (or more generally, if it is perfect) then these maps are bijections.

Let $A_{n+1}=(y_{ij})\in\cA_{n+1}$ be an alternating $(n+1)\times(n+1)$ matrix (so $y_{ii}=0$ and $y_{ij}=y_{ji}$) and define
\begin{gather*}
\beta\colon \ \cA_{n+1} \longrightarrow \cS_n,\qquad A_{n+1}\longmapsto \overline{A}_{n+1}:=\beta(A_{n+1}),\qquad
\big(\overline{A}_{n+1}\big)_{ij} := y_{ij} + y_{i,n+1}y_{j,n+1}.
\end{gather*}
For example,
\begin{gather*}
A_{4} = \begin{pmatrix}
 0&y_{12}&y_{13}&y_{14}\\
 y_{12}&0&y_{23}&y_{24}\\
 y_{13}&y_{23}&0&y_{34}\\
 y_{14}&y_{24}&y_{34}&0
 \end{pmatrix}
 \longmapsto
\overline{A}_4 =
\begin{pmatrix}
y_{14}^2&y_{12}+y_{14}y_{24}&y_{13}+y_{14}y_{34}\\
y_{12}+y_{14}y_{24}&y_{24}^2&y_{23}+y_{24}y_{34}\\
y_{13}+y_{14}y_{34}&y_{23}+y_{24}y_{34}&y_{34}^2
\end{pmatrix}.
\end{gather*}

It is not hard to verify that
\begin{gather*}
\beta(\alpha(S_n))_{ij} = (S_n)_{ij}^2,\qquad
\alpha(\beta(A_{n+1}))_{kl} = (A_{n+1})_{kl}^2,
\end{gather*}
for all $i,j=1, \ldots,n$ and all $k,l=1,\ldots,n+1$. Thus the maps $\beta\alpha\colon\allowbreak \cS_n\rightarrow \cS_n$ and $\alpha\beta\colon \cA_{n+1}\rightarrow \cA_{n+1}$ are the (coordinate wise) Frobenius maps on the respective vector spaces of matrices
\begin{gather*}
\beta\alpha = F_{\cS_n},\qquad \alpha\beta = F_{\cA_{n+1}}.
\end{gather*}

\subsection{From even to odd spinor varieties}\label{evenodd}
We denote the field of characteristic two by $K$. In Section \ref{pfmap} we considered an embedding $\cA_{n+1}\hookrightarrow\bS_{n+1}^+$, where $\bS_{n+1}^+$ parametrizes certain maximally isotropic subspaces for the quadratic form $q(y)=\sum\limits_{i=1}^{n+1}y_iy_{n+1+i}$ on $K^{2n+2}$. We define a hyperplane
\begin{gather*}
H\colon \ y_{n+1} + y_{2n+2} = 0 \qquad \big({\subset}\, K^{2n+2}\big).
\end{gather*}
The intersection $H\cap (q=0)$ can be identified with
the quadric in $K^{2n+1}$ defined by $q'$,
\begin{gather*}
q'=q_{|H}\colon \ K^{2n+1} \longrightarrow K,\qquad q'(z) = z_1z_{n+2} +\dots+ z_nz_{2n+1} + z_{n+1}^2,
\end{gather*}
simply by mapping $z=(z_1,\ldots,z_{2n+1})\mapsto y=(z_1,\ldots,z_{2n+1},z_{n+1})\in H$. A linear subspace contained in $q'=0$ has dimension at most $n$ and there is a unique family of such subspaces. If $W\subset (q=0)$ is a~maximal isotropic subspace for $q$, so $\dim W=n+1$, then $W':=W\cap H$ is a subspace of $q'=0$ of dimension $\geq n+1-1=n$ and we conclude that~$W'$ must have dimension~$n$, so $W'$ is maximally isotropic in $q'=0$. This sets up an isomorphism
\begin{gather*}
\bS_{n+1}^+ \stackrel{\cong}{\longrightarrow} \bS_{n+1}
\end{gather*}
between the spinor variety of ${\rm Spin}(2n+2)$ containing $W_0$ as in Section~\ref{isotr} and the spinor variety~$\bS_{n+1}$ of ${\rm Spin}(2n+1)$ that parametrizes the maximally isotropic subspaces for~$q'$.

\subsection[The geometry of $\beta$]{The geometry of $\boldsymbol{\beta}$}\label{geob}
We explain the geometry behind the map $\beta$, using the notation from Section~\ref{evenodd}. Since the field $K$ is assumed to have characteristic two, the alternating bilinear form $e'$ defined by $q'$ is degenerate
\begin{gather*}\begin{split}
& e'(z,w) := q(z+w) + q(z) + q(w),\\ & e'(z,w) = z_1w_{n+2}+\dots+z_nw_{2n+1} + z_{n+2}w_1+\dots+z_{2n+1}w_n,\end{split}
\end{gather*}
since the variables $z_{n+1}$, $w_{n+1}$ do not appear in~$e'$. More intrinsically, define the subspace
\begin{gather*}
\ker(e') := \big\{z\in K^{2n+1}\colon e'(z,w)=0 \ \forall\, w\in K^{2n+1} \big\},
\end{gather*}
then we see that $\ker(e')$ is one-dimensional (and is spanned by the $n+1$-st standard basis vector of~$K^{2n+1}$). We consider the quotient space $K^{2n+1}/\ker(e')\cong K^{2n}$ where we map $z=(z_1\ldots,z_{2n+1})\mapsto x=(z_1,\ldots, \widehat{z_{n+1}},\ldots,z_{2n+1})$, i.e., we omit the $(n+1)$-st coefficient. Since $e'$ is bilinear, it defines a non-degenerate alternating form~$\bar{e}$ on this quotient space simply by defining $\bar{e}(x,y):=e'(z,w)$, where $z,w\in K^{2n+1}$ map to $x,y\in K^{2n}$ respectively.

Let $W'$ be a maximally isotropic subspace in $q'=0$. Then for $z,w\in W'\subset (q'=0)$ we have $e'(z,w)=0$, and thus $\overline{W'}$, the image of $W'$ in $K^{2n}$, is an isotropic subspace for the symplectic form $\bar{e}$ on $K^{2n}$. Since $\ker(e')\cap (q'=0)=\{0\}$, the projection $\overline{W'}$ also has dimension~$n$ and hence~$\overline{W'}$ is a Lagrangian subspace for the symplectic form $\bar{e}$ on $K^{2n}$. Thus we have a map
\begin{gather*}
\beta\colon \ \bS_{n+1} \longrightarrow {\rm LG}(n,2n),\qquad W' \longmapsto \overline{W'}.
\end{gather*}
It is well-known that the orthogonal group ${\rm O}(q)$ of $q$ acts as the identity on $\ker(e')$ and that the projection to $K^{2n}$ induces a homomorphism (an isogeny) of algebraic groups ${\rm O}(q)\rightarrow {\rm Sp}(2n)$, the symplectic group defined by $\bar{e}$ (\cite[Section~4.11]{Steinberg}, \cite{Milne}, \cite[Section~7.1, Remark~7.1.6]{CGP}).

\begin{Proposition} \label{betas} The map $\beta\colon \bS_{n+1}\rightarrow {\rm LG}(n,2n)$ we just defined induces the map $\beta\colon \cA_{n+1}\rightarrow\cS_n$ from Section~{\rm \ref{mapas}}.
\end{Proposition}

\begin{proof} \looseness=-1 Given $Y\in \cA_{n+1}$, let $W_Y\subset K^{2n+2}$ be the subspace spanned by the columns of the $(2n+2)\times(n+1)$ matrix $M$ with blocks ${\bf I}$ and~$Y$. The intersection $W'_Y:=H\cap W_Y$ is spanned by the $n$ vectors $c_i+y_{i,n+1}c_{n+1}$, $i=1,\ldots,n$, where $c_i$ is the $i$-th column of $M$, in fact the $n+1$ and $2n+2$ coefficients of $c_i+y_{i,n+1}c_{n+1}$ are $0+y_{i,n+1}\cdot1$ and $y_{i,n+1}+y_{i,n+1}\cdot 0$ respectively, and their sum is indeed zero, showing that these vectors do lie in $H\cap W_Y$. Next we project these vectors to $K^{2n}$, so we omit the $(n+1)$-st coefficients, their span is then $\overline{H\cap W_Y}$. The image vectors are the columns of the $2n\times n$ matrix with blocks ${\bf I}$ and $\overline{Y}$, which proves that $\beta$ induces $Y\mapsto\overline{Y}$.\end{proof}

\begin{Remark}\label{remarkisog} The underlying reason for the results we obtained thus seems to be the isogeny of the linear algebraic groups ${\rm SO}(2n+1)\rightarrow {\rm Sp}(2n)$ (of type $B_n$ and $C_n$ respectively) over a field of characteristic two, cf.\ \cite[Section~4.11]{Steinberg}, \cite{Milne}, \cite[Section~7.1, Remark 7.1.6]{CGP}. The description of the isogeny leads directly to the map $\beta\colon \bS_{n+1}\rightarrow {\rm LG}(n,2n)$.

Using this isogeny, in \cite[p.~197]{Kleidman} one finds the definition of the ($2^n$-dimensional) Spin representation of ${\rm Sp}(2n)(K)$, where $K$ is a field of characteristic two. In Section~\ref{pluc} we recalled the decomposition of $\wedge^kV$ into ${\rm Sp}(2n)$-representations, where $V$ is the standard $2n$-dimensional representation of ${\rm Sp}(2n)$ in case the field $K$ has characteristic zero.

Now we assume that the field $K$ has characteristic two. The symplectic form on $V$ still induces an ${\rm Sp}(2n)$-equivariant contraction map $\partial_n\colon \wedge^nV\rightarrow \wedge^{n-2}V$ and $\big({\wedge}^nV\big)_0:=\ker(\partial_n)$ is an ${\rm Sp}(2n)$-subrepresentation of $\wedge^n V$ (but if $n>3$, then the dimension of $\big({\wedge}^nV\big)_0$ in characteristic two is larger than its dimension in characteristic zero, cf.\ \cite[Theorem~2.2]{Gow}). Gow~\cite{Gow} showed that now the image of the contraction map $\partial_{n+2}\colon \wedge^{n+2}V\rightarrow \wedge^nV$ has codimension $2^n$ in $\ker(\partial_n)$ and that the quotient ${\rm Sp}(2n)$-representation $\ker(\partial_n)/\im(\partial_{n+2})$ is the Spin representation of ${\rm Sp}(2n)$. The explicit description of the quotient module given in \cite[Section~3]{Gow} shows also that the composition ${\rm LG}(n,2n)\rightarrow\PP\big(\big({\wedge}^n V\big)_0\big)\rightarrow\PP(\ker(\partial_n)/\im(\partial_{n+2}))$ factors over the principal minor map $\underline{\pi}\colon {\rm LG}(n,2n)\rightarrow\PP^{2^n-1}$.

Since the composition $\underline{\pi}$ of the Pl\"ucker map with Gow's quotient map
\begin{gather*}
{\rm LG}(n,2n) \longrightarrow \PP\big({\wedge}^nV\big)_0 \longrightarrow \PP(\ker(\partial_n)/\im(\partial_{n+2})) = \PP^{2^n-1}
\end{gather*}
is equivariant for the action of ${\rm Sp}(2n)$ and ${\rm Sp}(2n)$ acts transitively on ${\rm LG}(n,2n)$, this composition is everywhere defined (so there are no base points). The image of ${\rm LG}(n,2n)$ is then the unique closed orbit of ${\rm Spin}(2n+1)$ in the projectivization of its spin representation, which is the spinor variety $\bS_{n+1}$. Thus we get a map ${\rm LG}(n,2n)\rightarrow \bS_{n+1}$ `for free'.

So, just from the isogeny and Gow's results, we recover a main result which we previously deduced from explicit computations. However, if one would like to know what the restriction of this map to the subset of ${\rm LG}(n,2n)$ parametrized by symmetric matrices is and what the quadratic relations between principal minors in characteristic two are, then the results of the first part of this paper are still useful.
\end{Remark}

\section{Freudenthal triple systems}

\subsection{Outline}
The case $n=3$, where we considered Cayley's hyperdeterminant and the Lagrangian Grassmannian ${\rm LG}(3,6)$, and the case $n=6$, where we considered the spinor variety $\bS_6\subset\PP^{31}$ associated to ${\rm Spin}(12)$, appear in the context of groups of type $E_7$~\cite{Brown} and reduced Freudenthal triple systems related to cubic Jordan algebras, as well as in the Freudenthal magic square~\cite{LM}. These subjects are also of relevance in Maxwell--Einstein four-dimensional supergravity (see, e.g., \cite{BFGM,BDFMR,CFMZ,FGK,FG,FM}), as well as in the so-called black-hole/qubit correspondence (cf.~\cite{BDDER,BDL,D,KL}).

It should be noticed that one usually excludes fields of characteristic $2$ and $3$ in these subjects. However, e.g., in~\cite{BDDR} integral Freudenthal triple systems and integral cubic Jordan algebras have been studied, and these can be reduced modulo two. Below we will also consider an approach to these subjects through the algebraic geometry of tangential varieties of certain homogeneous spaces.

\subsection[Groups of type $E_7$]{Groups of type $\boldsymbol{E_7}$}\label{jor}

A group of type $E_7$ is the subgroup, which we denote by $G_4$, of ${\rm GL}(R)$, where $R$ is a finite-dimensional vector space over a field $k$, which preserves a non-degenerate alternating form $e$ and a homogeneous quartic polynomial~$q$ (cf.~\cite{Brown,Kru}). It turns out that given some additional conditions, including a compatibility between~$e$ and~$q$ and $\operatorname{char}(k)\ne 2,3$, as well as the condition that $- \frac{1}{2} q(f,f,f,f)$ is a~non-zero square for some $f\in R$ (such triples $(R,q,e)$ are called reduced, \cite[p.~90]{Brown}), the vector space $R$ decomposes as
\begin{gather*}
R = k\oplus J \oplus J \oplus k,
\end{gather*}
where $J$ is a (cubic) Jordan algebra \cite{ja}, in \cite[Section~3.1]{Kru} $R$ is denoted by~$M(J)$. The Jordan algebras of interest for us are those given by the $3\times 3$ matrices of the form
\begin{gather*}
A := \begin{pmatrix} a&z&\bar{y}\\\bar{z}&b&x\\y&\bar{x}&c\end{pmatrix},\qquad a,b,c\in k,\qquad x,y,z\in C,
\end{gather*}
where \looseness=1  $C$ is a composition algebra with involution $x\mapsto \bar{x}$. and norm $n\colon C\rightarrow k$. The best known examples are $k=\RR$ and $C=\RR,\CC,\HH,\OO$, where $\HH$ are the quaternions and $\OO$ are the octonions. For simplicity we will also consider the corresponding split algebras as in \cite[Section~2.1]{Kru} here.

\looseness=1 Notice that the dimension of $J$ as a $k$-vector space is $3q+3$ where $q:=\dim C$ and $\dim R=2+2(3q+3)=6q+8$. Over the complex numbers (the split case), the group~$G_4$ will be ${\rm SL}(2,\CC)^3$, ${\rm Sp}(6,\CC)$, ${\rm SL}(6,\CC)$, ${\rm Spin}(12,\CC)$ and $E_7(\CC)$ for $q=0,1,2,4,8$, respectively, cf.\ \cite[Table~3]{Holweck}. The $G_4$-orbit $Y_q$ of the highest weight vector in $\PP R=\PP^{6q+7}$ is a complex projective algebraic variety of dimension~$3q+3$, it is the unique closed orbit of~$G_4$ in~$\PP R$. The tangential variety~$X_q$ of~$Y_q$ (see also Section~\ref{tgv}) has dimension $6q+6$ and is defined by the quartic polynomial~$q$ (cf.~\cite{FGK,Holweck,LM}).

\subsection{The alternating form and the quartic}\label{jor2}
The algebra $J$ comes with a norm $N\colon J\rightarrow k$, which is homogeneous of degree three, and which generalizes the determinant of a $3\times 3$ matrix:
\begin{gather*}
N(A) = abc - ax\bar{x} - by\bar{y} - cz\bar{z} + (xy)z + \bar{z}(\bar{y}\bar{x}).
\end{gather*}
One defines the regular trace $\operatorname{Tr}\colon J\rightarrow k$ and a symmetric bilinear form $(-,-)\colon J\times J\rightarrow k$ by
\begin{gather*}
\operatorname{Tr}(A) = a+b+c,\qquad (A,B) := \operatorname{Tr}\left( \tfrac{1}{2} (AB+BA)\right).
\end{gather*}
Finally there is `sharp' operation on $J$, similar to the adjoint of a matrix,
\begin{gather*}
A^\sharp := \begin{pmatrix}
bc-n(x)&\bar{y}\bar{x}-cz&zx-b\bar{y}\\xy-c\bar{z}&ac-n(y)&\bar{z}\bar{y}-ax\\
\bar{x}\bar{z}-by&yz-a\bar{x}&ab-n(z)\end{pmatrix}.
\end{gather*}
We will write elements of $R$ as four tuples $(a,A,B,b)$ with $a,b\in k$ and $A,B\in J$ but often a~matrix notation is used (cf.\ \cite[equation~(29)]{Kru}). With these definitions (cf.\ \cite[Section~2.4]{Kru}), the alternating form $e\colon R\times R\rightarrow k$ is given by (cf.\ \cite[Section~3.1]{Kru})
\begin{gather*}
e((a,A,B,b),(c,C,D,d)) = ad - bc+(A,D) - (B,C),
\end{gather*}
and the quartic form $q\colon R\rightarrow k$ is
\begin{gather*}
G := -2G',\qquad G'((a,A,B,b)) = -4\big(A^\sharp,B^\sharp\big) + 4aN(A) + 4bN(B) + \big((A,B)-ab\big)^2.
\end{gather*}
A Freudenthal triple system is a vector space $R$ with a non-degenerate alternating form and a~composition $T\colon R\times R\times R\rightarrow R$ such that certain conditions are satisfied cf.\ \cite[Section~3.1]{Kru}. Under additional conditions one recovers a group of type $E_7$ as the automorphism group of a~Freudenthal triple system.

\begin{Example}[$q=0$] \label{q=0}
Recall that $q:=\dim C$ and if $C=0$ then $J$ is the three-dimensional algebra of diagonal matrices and $\dim R=8$. We change the coordinates as follows
\begin{gather*}
x:= (a,\operatorname{diag}(a_{11},a_{22},a_{33}),\operatorname{diag}(b_{11},b_{22},b_{33}),b) =(z_{000},z_{110},z_{101},z_{011},z_{001},z_{010},z_{001},z_{111}),
\end{gather*}
and similarly $((c,C,D,d)=(y_{000},\ldots,y_{111})$. The symplectic form is then
\begin{gather*}
e(x,y) =z_{000} y_{111} - z_{001} y_{110} - z_{010} y_{101} + z_{011} y_{100} - z_{100} y_{011} + z_{101} y_{010} \\
\hphantom{e(x,y) =}{}+ z_{110} y_{001} - z_{111} y_{000}
\end{gather*}
and the quartic polynomial $G'$ is the hyperdeterminant: $G'(x)=H$ (cf.\ Example~\ref{n3}). Since the hyperdeterminant becomes a square modulo two, as we observed in Example~\ref{n3}, we consider some other cases of the construction above.
\end{Example}

\subsection{Reduction modulo two}
More generally, if we assume that $N(A),N(B)$ and $(A,B)$ are polynomials in the coefficients of~$A$,~$B$ with integer coefficients, then the reduction of the quartic form $G'$ is
\begin{gather*}
G'\equiv Q^2\quad\text{mod}\, 2,\qquad Q((a,A,B,b)) := (A,B)-ab,
\end{gather*}
so it becomes the square of a quadratic polynomial. This may be particularly relevant when considering integral Freudenthal triple systems in characteristic~2, since in this case the so-called Freudenthal duality is always defined, albeit becoming simply an antinvolutive electric-magnetic symplectic duality transformation~\cite{BDDR,FKM}.

Since $(A,B)$ is symmetric and bilinear in $A$, $B$ the bilinear form associated to $(A,B)-ab$, defined as $Q(x+y)-Q(x)-Q(y)$, is
\begin{gather*}
\big((A+C,B+D)-(a+c)(b+d)\big) - \big((A,B)-ab\big) - \big((C,D)-cd)\big)\\
\qquad{} = (A,D) + (B,C) - ad - bc,
\end{gather*}
hence (notice that `$+$'=`$-$') the associated bilinear form is now the alternating form $e$. In general, the associated bilinear form $e$ of a quadratic form over a field of
characteristic two is alternating (that is $e(x,x)=0$) which follows easily by putting $x=y$ in $Q(x+y)-Q(x)-Q(y)=e(x,y)$.

In particular, the group fixing both the alternating form and the quadric is just the orthogonal group fixing the quadric, see also \cite[Remark~20]{Kru}. The case $q=1$ presents some extra features since in that case the quadric and the alternating form are degenerate, see Example~\ref{q1}.

\begin{Example}[$q=1$]\label{q1} We consider the constructions from Section~\ref{jor2} for the case that $k=C=\RR$, with $\bar{x}=x$ and $n(x)=x$ for $x\in C$. Then $q=\dim_kC=1$ and $\dim J=6$, $\dim R=14$ and $R=\big({\wedge}^3k^6\big)_0$, which is an irreducible ${\rm Sp}(6,\RR)$-representation, cf.\ Section~\ref{pluc}. In that case $J$ is just the six-dimensional $\RR$-algebra of $3\times 3$ symmetric matrices, $N(A)=\det(A)$ and $A^\sharp$ is the adjoint matrix of $A$ and we write $A=(a_{ij})$, $\ldots$, $D=(d_{ij})$ where $A,\ldots,D$ are $3\times 3$ symmetric matrices. Then
\begin{gather*}
e(((a,A,B,b),(c,C,D,d))=a_{11}d_{11 }- b_{11}c_{11} + a_{22}d_{22 }- b_{22}c_{22} +a_{33}d_{33 }- b_{33}c_{33} \\
\qquad{} +2\bigl(a_{12}d_{12 }- b_{12}c_{12} + a_{13}d_{13 }- b_{13}c_{13} + a_{23}d_{23} - b_{23}c_{23 }) + ad - bc,
\end{gather*}
notice the factors $2$ which appear. The quartic $G'(x)=a_{11}^2b_{11}^2 + 4a_{11}a_{22}a_{33}a+\dots+a^2b^2$ has $44$ terms. Reducing modulo two one finds
\begin{gather*}
G' \equiv Q^2,\qquad Q := a_{11}b_{11} + a_{22}b_{22} + a_{33}b_{33} + ab,
\end{gather*}
and $Q=0$ is a singular quadric in $\PP^{13}$ (notice that the $6=3+3$ variables $a_{ij}$, $b_{ij}$ with $i<j$ do not appear). The associated bilinear form of $Q$ is the reduction of $e$ mod 2, and this is a~degenerate alternating form. Notice that while~\cite{Kru} discusses the next three cases, $q=2,4,8$, the case $q=1$ is avoided.

As we observed in Remark~\ref{remarkisog}, Gow showed that the ${\rm Sp}(6)$-representation $R$ has a subrepresentation $\im(\partial_{5})$, the image of the contraction map $\partial_5\colon \wedge^5 k^6 \rightarrow\wedge^3k^6$, with $2^3=8$-dimensional quotient $R/\im(\partial_5)$, the Spin representation of ${\rm Sp}(6)$. In this case, the subrepresentation coincides with the kernel of the bilinear form defined by~$Q$, written again as~$Q$, that is
\begin{gather*}
\im(\partial_{5}) = \{ v\in R\colon Q(v,w)=0\ \forall\, w\in R \}.
\end{gather*}
Thus $Q$ restricts to a nondegenerate quadratic form on the Spin representation of ${\rm Sp}(6)$. According to \cite[Section~4]{Gow}, such an ${\rm Sp}(2n)$-invariant non-degenerate quadratic form on the Spin representation exists for any~$n$.
\end{Example}

\subsection[The cases $q=2,4,8$]{The cases $\boldsymbol{q=2,4,8}$}\label{q8}
The cases $q=2,4,8$ do not seem to present special features. In case $q=8$ (see also~\cite{wilson}), we use the expression for the quartic invariant that we found in~\cite{freudenthal}. Let $X:=(x_{ij})$, $Y:=(y_{ij})$ be alternating $8\times 8$ matrices over the real or complex numbers. We view the $x_{ij},y_{ij}$ as coordinates on a $28+28=56$-dimensional vector space and we define a symplectic form on this vector space by requiring that the $x_{ij}$, $y_{ij}$ are the coordinates on a symplectic basis. The quartic invariant of~$E_7$ is defined as
\begin{gather*}
J := \operatorname{Pf}(X) + \operatorname{Pf}(Y) - \tfrac{1}{4} \operatorname{Tr}(XYXY) + \tfrac{1}{16} (\operatorname{Tr}(XY))^2.
\end{gather*}
The degree four polynomial $J$ has $1036$ terms and coefficients in $\big\{{\pm}1, \pm \frac{1}{2}, - \frac{1}{4} \big\}$. To find a~reduction mod $2$ we simply multiply $J$ by $4$ and then reduce mod $2$ to obtain a quartic $\overline{J}$ which has only~$28$ terms
\begin{gather*}
\overline{J} = x_{12}^2y_{12}^2 + x_{13}^2y_{13}^2 + \dots + x_{78}^2y_{78}^2 = \sum_{1\leq i<j\leq 28} x_{ij}^2y_{ij}^2,
\end{gather*}
which is the square of the quadratic polynomial $Q=\sum\limits_{i<j} x_{ij}y_{ij}$. Notice that the alternating form defined by $Q$ is indeed the symplectic form defined above.

\subsection{Tangential varieties}\label{tgv}
The representation of the group $G_4$ on $\PP R=\PP^{6q+7}$ has a unique closed orbit $Y_q$ of dimension $3q+3$ and one recovers the zero locus of the quartic invariant
as the $(6q+6)$-dimensional tangential variety $X_q$ of $Y_q$. This allows one to determine the equation of~$X_q$ over a field of characteristic two. In the case $q=1$ we again find the singular quadric from Example~\ref{q1} and in the cases $q=0,2,4$ where we did the computations, we again find a smooth quadric as before. We didn't attempt to compute the $q=8$ case.

Given a subvariety $Z$ of $\PP^N$, its tangential variety, often denoted by~$\tau(Z)$, is the union of all of its projective tangent spaces, equivalently, it is the union of all embedded tangent lines to~$Z$, see~\cite{Z}. If locally we have a parametrization
\begin{gather*}
\phi\colon \ k^n \longrightarrow Z,\qquad \phi(x) = (f_0(x):\ldots:f_N(x))
\end{gather*}
for \looseness=-1 certain functions on $k^n$, then the tangent spaces to the image of $\phi$ are locally parametrized by
\begin{gather*}
\tilde{\phi}\colon \ k^n\times k^n \longrightarrow \PP^N,\qquad (x,y) \longmapsto \left(\ldots: f_k(x)+\sum_{j=1}^n \frac{\partial f_k}{\partial x_j} (x)y_j:\ldots\right).
\end{gather*}
In particular, given an explicit local parametrization of $Z$ with dense image, one can compute the homogeneous polynomials vanishing on the tangential variety of~$Z$. For references and some recent results, mostly over the complex numbers, see~\cite{tangseg}.

As we will see in the next section, over a field of characteristic two the tangential varieties behave rather differently from other characteristics (in all our examples we find quadrics instead of quartics).

\looseness=1 One of the referees for this paper offered the following explanation (actually we simplify quite a bit) for this `bad' behaviour. In the examples we consider the secant variety of $Y_q$, which is the closure of the union of all lines in $\PP R$ joining two distinct points in~$Y_q$, is all of~$\PP R$. Moreover, a general point of $\PP R$ lies on a unique secant of~$Y_q$. Over an open subset $U\subset\PP R$ we thus have a~fibration $L\rightarrow U$ where the fiber~$L_r$ over $r\in U$ is the unique secant containing $r$.
On each secant there are the two points where $L_r$ intersects $Y_q$. These points give a subvariety $V\subset L$ and the induced map $\pi\colon V\rightarrow U$ has degree two. The map $\pi$ ramifies exactly when $L_r$ is tangent to $Y_q$, that is when $r\in X_q$, the tangential variety of~$Y_q$. Notice that locally $V$ is defined by an equation $a(r)s^2+b(r)st+c(r)t^2$ with $r\in U$ and $(s:t)$ homogeneous coordinates on the line~$L_r$. The ramification locus of $\pi$ is then defined by $b^2(r)-4a(r)c(r)=0$
which reduces to $b(r)^2=0$ when the field has characteristic two. The examples in the next section in fact show that the quartics defining the tangential variety $X_q$ are the squares of quadrics in that case. A worked out example for a related, but simpler, case $q=-1$ is given in \cite[Example~4.7]{Russo}.

\subsection{The various cases} \label{varcas}
We verified in the cases $q=0,1,2,4$ that the zero locus $X_q$ of the quartic invariant, which is the tangential variety of the closed orbit $Y_q$, is actually a quadric in characteristic two, by explicitly finding an equation for the tangential variety to $Y_q$ in characteristic two.

In case $q=0$, we verified that the quadric we found in Example~\ref{q=0} is the tangential variety of the Segre 3-fold $Y_0$, the image of $\PP^1\times\PP^1\times\PP^1$ in~$\PP^7$.

For $q=1$, one finds by direct computation that the image $Y_1$ of ${\rm LG}(3,6)$ under the Pl\"ucker map spans a $\PP^{13}\subset\PP^{19}$. The tangential variety $X_1$ of ${\rm LG}(3,6)$ in this $\PP^{13}$ is a singular quadric of rank $8$. The singular locus of the quadric is a $\PP^5\subset\PP^{13}$. Notice that this $\PP^5$ is mapped into itself under the action of ${\rm Sp}(6)$ on ${\rm LG}(3,6)$ and $\PP^{13}$, in fact $\PP^5$ is $\PP(\im\partial_5)$, the subrepresentation of $\PP R$ as in Example~\ref{q1}.

For $q=2$, we found that the tangential variety of the Pl\"ucker embedded $Y_2={\rm Gr}(3,6)\subset\PP^{19}$ is a smooth quadric.

Also for $q=4$, we checked that the tangential variety of the 15-dimensional $Y_4=\bS_6$, Pfaffian embedded in $\PP^{31}$, is a smooth quadric (cf.~\cite{M} for such aspects of the geometry of spinor varieties).

In case $q=8$, the variety $Y_8$ of dimension $27$ in $\PP^{55}$ is known as the Freudenthal variety~\cite{MM}. Its tangential variety $X_8$, over the complex numbers, is the quartic hypersurface defined by $J=0$ in Section~\ref{q8}. We haven't computed what happens over a field of characteristic two since the only parametrization of $Y_8$ that we know of is rather cumbersome.

\subsection*{Acknowledgments}
BvG would like to thank L.~Oeding and W.~van der Kallen for helpful correspondence and discussions. We are indebted to the referees of this paper for comments and suggestions
for improvements.

\pdfbookmark[1]{References}{ref}
\LastPageEnding

\end{document}